\documentclass[a4article, english]{lipics}                        
\usepackage[utf8]{inputenc}   

\usepackage{amsmath}
\usepackage{amsthm}
\usepackage{amssymb}
\usepackage{dsfont}

\usepackage{algorithmicx}
\usepackage{algorithm}
\usepackage{algpseudocode}

\theoremstyle{plain}
\newtheorem{thm}{Theorem}[section]
\newtheorem{lem}[thm]{Lemma}
\newtheorem{prop}[thm]{Proposition}

\newtheorem*{lemappx}{Lemma}
\newtheorem*{thmappx}{Theorem}

\theoremstyle{definition}
\newtheorem{defn}{Definition}[section]

\theoremstyle{remark}
\newtheorem*{rem}{Remark}

\newcommand{\opt}{\text{OPT}}
\newcommand{\cost}{\text{Cost}}
\newcommand{\eps}{\varepsilon}
\newcommand{\argmax}{\text{argmax}}
\newcommand{\argmin}{\text{argmin}}

\newcommand{\lemThreshold}{1/\eps^2 + 1/\eps}

\newcommand{\len}{\text{length}}
\newcommand{\dist}{\text{dist}}
\newcommand{\tsp}{\text{TSP}}
\newcommand{\steiner}{\text{ST}}

\newcommand{\ins}{\text{In}}
\newcommand{\cross}{\text{Cross}}
\newcommand{\out}{\text{Out}}

\begin{document}

\keywords{Local Search, PTAS, Facility Location, $k$-Median, TSP, Steiner Tree}

\title{Effectiveness of Local Search for Geometric Optimization}

\author[1]{Vincent Cohen-Addad}
\author[1]{Claire Mathieu\footnote{Partially supported by ANR RDAM}}
\affil[1]{D\'{e}partement d’Informatique, UMR CNRS 8548,
\'{E}cole Normale Sup\'{e}rieure, Paris, France\\
\{vcohen, cmathieu\}@di.ens.fr}

\maketitle
\begin{abstract}
What is the effectiveness of local search algorithms for geometric problems in the plane?
We prove that  local search with neighborhoods of magnitude $1/\epsilon^c$ is an approximation scheme for  the following problems in the Euclidean plane: TSP with random inputs, Steiner tree with random inputs, uniform facility location (with worst case inputs), and bicriteria $k$-median (also with worst case inputs). 
The randomness assumption is necessary for TSP.

\end{abstract}

\section{Introduction}

{\bf Local search.} Local search techniques are popular heuristics for hard combinatorial optimization problems. Given a feasible solution, the algorithm repeatedly performs operations from the given class, each improving the cost of the current solution, until a solution is reached for which no operation yields an improvement (a locally optimal solution). Alternatively, we can view this as a neighborhood search process, where each
solution has an associated neighborhood of adjacent solutions, i.e., those that can be reached with a
single operation, and one moves to a better neighbor until none.
Such techniques are easy to implement, easy to parallelize, and fast and give good results. 
One advantageous feature of local search  algorithms is their flexibility; they can be applied to arbitrary cost functions, even in the presence of
additional constraints.
However, there has long been a gap between worst-case guarantees and real-world experience.
Thus, it is interesting to analyze such algorithms rigorously and, even in settings where alternative, theoretically optimal polynomial-time algorithms are known.

{\bf Problems studied.} We focus on Euclidean problems in the plane (the results extend to small dimensions), and study clustering and network connectivity type problems: the traveling
salesman problem (TSP), Steiner tree, facility location, and k-median.
The \textit{traveling salesman} problem is to connect $n$ input points with  a tour of minimum total length. The \textit{Steiner tree} problem, given $n$ terminal points, is to choose additional {\em Steiner} points so as to minimize the length of the minimum tree spanning terminal and Steiner points. The \textit{facility location} problem, given $n$ client points and a facility opening cost $f$, chooses how many facilities to open and where to open them to minimize the combination of the cost of opening facilities and of the total distance from each client to the nearest open facility. The \textit{$k$-median} problem, given $n$ points and an integer $k$, chooses where to open $k$ facilities so as to minimize the total distance from each client to the nearest open facility.

{\bf Algorithms.} Our goal is to prove, under minimal assumptions, that local search finds solutions whose cost is within a $(1+\epsilon)$ factor of optimal. For that goal, local search must do a little more: instead of modifying the current solution by swapping a single point, edge or edge pair (depending on the problem) in and out of the solution, our version of local search swaps up to $1/\epsilon^c$ points, edges or edge pairs. This is a standard variation of local search (particularly for the traveling salesman tour), whereby each iteration is slowed down due to an increase in the size of the neighborhood, but the local optimum tends to be reached after fewer iterations and is of higher quality. Moreover, most implementations of local search do not continue iterating all the way to a local optimum, but stop once the gain obtained by each additional iteration is essentially negligible. Our algorithm thus has a stopping condition, when no local exchange could improve the cost by more than a factor of $1-1/n$. Then, the runtime is polynomial, at most $n^{1/\epsilon^{O(1)}}$.

{\bf Results.} Our results are as follows.
\begin{enumerate}
\item
 For TSP, we  assume that the input points are random uniform in $[0,1]^2$. Here
local search swaps $O(1/\epsilon^c)$ edges in the tour. Then  local
search finds a solution with cost $(1+\mathcal{O}(\epsilon)) OPT$.
The proof is not difficult and serves as a warm-up to the later sections. The random input
assumption is necessary : in the worst-case setting, we give an
example where a locally optimal solution has cost more than
$(2-\epsilon)OPT$. 

\item 
Similarly, for Steiner tree, assuming random uniform input, again
local search  finds a solution with cost $(1+\epsilon) OPT$.
\item
 For  facility location, we prove the following: consider the
version of local search where local moves consist of adding, deleting
or swapping $O(1/\epsilon^c)$ facilities. Then, even for worst case inputs, local search finds a
solution with cost $(1+\epsilon) OPT$. This is the core result of the paper.
We transform the dissection technique from Kolliopoulos and Rao~\cite{Kolliopoulos07} into a tool for analyzing local search.
\item
 For $k$-median, our result is similar, except that local search uses
$(1+\epsilon)k$ medians instead of $k$, so that result is bicriteria. This is a technical, variant of the facility location result.
\end{enumerate}

{\bf Related work.}
\textbf{TSP and Steiner Tree.}
The TSP problem in the Euclidean plane has a long history, including work with local search~\cite{Croes58,Lin65,Lin73}. Most relevant is the work of Karp~\cite{Karp77} giving a simple construction of a near-optimal tour when points are drawn from a random distribution. That work has been subsumed by the approximation schemes of Arora~\cite{Arora96} (and its improvements~\cite{Arora97,Rao98}) and of Mitchell~\cite{Mitchell99}, using a hierarchical dissection technique. Arora noted the relation between that technique and local search, observing: 
\begin{quote}
\textit{
Local-exchange algorithms for the TSP work by identifying possible edge exchanges in the current tour that lower the cost 
[…].
Our dynamic programming algorithm can be restated as a slightly more inefficient backtracking […]. 
Thus it resembles $k$-OPT for $k=O(c)$, except that cost-increasing exchanges have to be allowed in order to undo bad guesses. Maybe it is closer in spirit to more ad-hoc heuristics such as genetic algorithms, which do allow cost-increasing exchanges.
}
\end{quote}

In fact, even with neighborhoods of size $f(\epsilon)$, even in the Euclidean plane, local search for TSP can get stuck in a local optimum whose value is far from the global optimum (See Fig. \ref{fig:lb_tsp}). However, in the case of random inputs the intuition is correct.
Local search algorithms have been widely studied for TSP, but mostly for either a local neighborhood limited to size of 2 or 3 (the 2-OPT or 3-OPT algorithms), or for the general metric case. Those studies lead to proofs of constant factor approximations, see~\cite{Chandra94,Johnson97,Mersmann12,Lin73,Rosenkrantz77}. In particular, in~\cite{Chandra94}, it is proved (by example) that  for Euclidean TSP 2-OPT cannot be a constant-factor approximation in the worst case.
For the metric Steiner Tree problem, the best approximation algorithm up to 2010 was a constant factor approximation due to Robins and Zelikovsky and was by local search \cite{Robins00}.

\textbf{Facility Location and $k$-Median.}
For clustering problems — facility location and $k$-median — there has also been much prior work. A proof of NP-Hardness of $k$-median even in the Euclidean setting is given in~\cite{Megiddo84}. 
The first theoretical guarantees for local search algorithms for 
clustering problems are due to Korupolu et al. \cite{Korupolu00}. They show that the local search algorithm which allows swaps of size $p$ is a constant factor approximation for the metric case of the $k$-Median and Facility Location problems. However, for $k$-Median the algorithm requires a constant-factor
blowup in the parameter $k$.
By further refining the analysis, Charikar et al. \cite{Charikar05} improved the 
approximation ratio.
More recently, Arya et al. showed in \cite{Arya04} that the local search
algorithm which allows swaps of size $p$ is a $3+2/p$-approximation 
without any blowup in the number of medians.
Nevertheless, no better results were known for the Euclidean 
case (See the survey paper \cite{Vygen05}).
Kolliopoulos and Rao define in \cite{Kolliopoulos07} a recursive ``adaptive'' dissection of a square enclosing the input points. At each dissection step
\footnote{There is also a ``sub-rectangle'' step not described here.},
they cut the longer side of each rectangle produced by the previous step in such a way that each of the two parts has roughly the same surface area.
Our analysis uses a new version of their dissection algorithm to analyze the local search algoritm.

\textbf{Other related work.}
The question of the efficiency of local search for Euclidean problems was already posed by Mustafa and Ray and Chan and Har-Peled. They proved that local search (with local neighborhood enabling moves of size $\Theta(1/\epsilon)$) gives approximation schemes for hitting circular disks in two dimensions with the fewest points, for several other Euclidean hitting set problems~\cite{Mustafa09}, and for independent sets of pseudo-disks~\cite{Chan09}.
This led to further PTASs by local search for dominating set in disks graph \cite{Gibson10} and for terrain guarding \cite{Krohn14}.
Those papers rely on the combinatorial properties of bipartite planar graphs.
Our analysis technique is different since we rely on dissections.

One problem related to facility location is $k$-means. For $k$-means, Kanungo, Mount, Netanyahu and Piatko~\cite{Kanungo04} proved that local search gives a constant factor approximation. Much remains to be understood.

We also note that there exists proofs of constant factor approximation by local search for the metric capacitated facility location \cite{Chudak05}.

\textbf{Plan.}
The paper is organized as follows: in the next section, as a warm-up we prove the results on TSP and Steiner tree for random inputs. We then analyze local search for facility location, proposing a new recursive dissection. We suitably extend lemmas from~\cite{Kolliopoulos07}. The meat of that section is the proof of Proposition
\ref{thm:new_struct}, which is our main technical contribution.
We end with the $k$-median result, that requires additional ideas to deal with the cardinality constraint.

\section{Polynomial-Time Local Search Algorithms}
Throughout this paper, we denote by $L \bigtriangleup L'$ the symmetric difference of the sets $L$ and $L'$.
We present the local search algorithm that is considered in this paper (see Algorithm \ref{algo:LS} below).

\begin{algorithm}
  \caption{Local Search ($\eps$)}
  \label{algo:LS}
  \begin{algorithmic}[1]
    \State \textbf{Input:} A set $\mathcal{C}$ of points in the Euclidean plane
    \State $S \gets$ Arbitrary feasible solution (of cost at most $\mathcal{O}(2^{n} \opt)$).
    \While{$\exists$ $S'$ s.t. Condition($S',\eps$) 
      \textbf{and} cost($S'$) $\le$ $(1-1/n)$ cost($S$)\\}
    
    \State $S \gets S'$
    \EndWhile
    \State \textbf{Output:} $S$     
  \end{algorithmic}

\end{algorithm}

Note that the type of $S$, Condition, $f(\eps)$ and $\cost(S)$ are problem dependent. Namely,
\begin{itemize}
\item for Facility Location, $S$ is a set of points, Condition($S',\eps$) is $ |S \bigtriangleup S'| = \mathcal{O}(1/\eps^3)$ and 
  $\cost(S) = |S| + \sum\limits_{c \in \mathcal{C}} \min\limits_{s \in S} d(c,s)$;
\item for $k$-Median, $S$ is a set of points, Condition($S',\eps$) is $|S \bigtriangleup S'| = \mathcal{O}(1/\eps^9)$ and $|S'| \le (1+3\eps)k$
  and $\cost(S) = \sum\limits_{c \in \mathcal{C}} \min\limits_{s \in S} d(c,s)$;
\item for TSP $S$ is a set of edges, Condition($S',\eps$) is $|S \bigtriangleup S'| = \mathcal{O}(1/\eps^2)$  and ``$S'$ is a tour and there
  is no two edges intersecting'' (if the
  initial tour contains intersecting edges we start by modifying the tour so that no two edges intersect)
  and $\cost(S) = \sum\limits_{s \in S} \len(s)$;
\item for Steiner Tree, $S$ is a set of points, Condition($S',\eps$) is $|S \bigtriangleup S'| = \mathcal{O}(1/\eps^2)$ and
  $|S'| \le n$ (if the initial set of Steiner vertices is greater than $n$, we greedily remove Steiner vertices until the set has size $n$)
  and $\cost(S) = \text{MST}(S \cup \mathcal{C})$, where $\text{MST}(S \cup \mathcal{C})$ is the length of the minimum spanning tree of the points in
  $S \cup \mathcal{C}$.
\end{itemize}

We now focus on the guarantees on the execution time of the algorithms presented in this paper.
The proof of the following Lemma is deferred to Appendix \ref{appx:ptime}.

\begin{lem}
\label{lem:polytime}
  The number of iterations of Algorithm \ref{algo:LS} is polynomial for the Facility Location, $k$-Median, Traveling Salesman and Steiner Tree Problems.
\end{lem}

\begin{rem}
  Up to discretizing the plane and replacing $(1-1/n)$ by $(1-\Theta(1/n))$, finding $S'$ takes time $\mathcal{O}(n^{\mathcal{O}(1/\eps^c)}\eps^{-1})$, for some constant $c$ which depends on the algorithm.
\end{rem}

\section{Euclidean Traveling Salesman Problem and Steiner Tree}


\begin{thm}\label{thm:ptas_expct}
  Consider a set of points chosen independently and uniformly in $[0,1]^2$. Algorithm \ref{algo:LS} produces:
  \begin{itemize}
  \item In the case of the Traveling Salesman problem, a tour whose length is at most $(1 + \mathcal{O}(\eps)) T_{\opt} $,
    where $T_{\opt}$ is the length of the optimal solution.
  \item In the case of the Steiner Tree problem, a tree  whose length is at most $(1 + \mathcal{O}(\eps)) T_{\opt} $,
    where $T_{\opt}$ is the length of the optimal solution.
  \end{itemize}

\end{thm}

To prove Theorem \ref{thm:ptas_expct}, we first prove the following result.

\begin{thm}\label{thm:ub_tsp_st}
  Consider an arbitrary set of points in $[0,1]^2$. Algorithm \ref{algo:LS} produces:
  \begin{itemize}
  \item In the case of the Traveling Salesman problem, a tour whose length is at most $(1 + \mathcal{O}(\eps^2))T_{\opt} + \mathcal{O}(\eps \sqrt{n})$,
    where $T_{\opt}$ is the length of the optimal solution.
  \item In the case of the Steiner Tree problem, a tree  whose length is at most $(1 + \mathcal{O}(\eps^2)) T_{\opt} + \mathcal{O}(\eps \sqrt{ n})$,
    where $T_{\opt}$ is the length of the optimal solution.
  \end{itemize}
\end{thm}

We model a random distribution of points in a region $\mathcal{P}$ of the plane by a
two-dimensional Poisson distribution $\Pi_n(\mathcal{P})$. The distribution $\Pi_n(\mathcal{P})$ is determined
by the following assumptions:
\begin{enumerate}
\item the numbers of points occurring in two or more disjoint sub-regions are
distributed independently of each other;
\item the expected number of points in a region $A$ is $n v(A)$ where $v(A)$ is the area
of $A$; and
\item as $v(A)$ tends to zero, the probability of more than one point occurring in $A$
tends to zero faster than $v(A)$.
\end{enumerate}
From these assumptions it follows that
Pr$[A \text{ contains exactly } m \text{ points}] = e^{- \lambda}\lambda^m/m!$, where $\lambda = nv(A)$.
The following result is known.

\begin{thm}
  \label{thm:random_tsp}
  \cite{Beardwood59}
  Let $\mathcal{P}$ be a set of $n$ points distributed according to a two-dimensional Poisson distribution $\Pi_n(\mathcal{P})$ in $[0,1]^2$ and
  let $T_n(\mathcal{P})$ be the random variable that denotes the length of the shortest tour
  through the points in $\mathcal{P}$.
  There exists a positive constant $\beta$ (independent of $\mathcal{P}$) such that
  $T_n(\mathcal{P})/\sqrt{n} \rightarrow \beta$ with probability 1.
\end{thm}

Assuming Theorems \ref{thm:ub_tsp_st} and \ref{thm:random_tsp}, we can prove Theorem \ref{thm:ptas_expct}.

\begin{proof}[Proof of Theorem \ref{thm:ptas_expct}]
  We focus on the Traveling Salesman case.
  Let $L$ be the tour produced by Algorithm \ref{algo:LS} and $T_{\opt}$ be the optimal tour.
  By Theorem \ref{thm:random_tsp}, we have that $\cost(T_{\opt}) = \mathcal{O}(\sqrt{n})$ with probability 1.
  Hence, Theorem \ref{thm:ub_tsp_st} implies 
  $$(1-\eps^2)\cdot \cost(L) \le \cost(T_{\opt}) + \mathcal{O}(\eps\sqrt{n}) = (1 + \mathcal{O}(\eps)) \cdot \cost(T_{\opt}).$$

  We now consider the random variable $ST_n(\mathcal{P})$ that denotes the length of the shortest Steiner Tree through the points in $\mathcal{P}$.
  Since the length of the optimal Steiner Tree is at least half the length of the optimal Traveling Salesman Tour,
  Theorem \ref{thm:random_tsp} implies that there exists a constant $\delta$ such that $ST_n(\mathcal{P})/\sqrt{n} \ge \delta$ with probability 1.
  Then, the exact same reasoning applies to prove the Steiner Tree case.
\end{proof}

The rest of the section is dedicated to the proof of Theorem \ref{thm:ub_tsp_st}.
To this aim,
we define a recursive dissection of the unit square according to a set of points $\mathcal{P}$.
At each step we cut the longer side of each rectangle produced by the previous step in such a way
that each of the two parts contains half the points of $\mathcal{P}$ that lie in the rectangle.
The process stops when each rectangle contains $\Theta (1/\eps^2)$ points of $\mathcal{P}$.
We now consider the final rectangles and we refer to them as \emph{boxes}.
Let $\mathcal{B}$ be the set of boxes.

\begin{lem}
  \label{lem:Karp77}
  \cite{Karp77} $\sum\limits_{b \in \mathcal{B}} |\partial b| = \mathcal{O}(\eps \sqrt{ |\mathcal{P}|})$, where $|\partial b|$ is the perimeter
  of box $b$ and $|\mathcal{P}|$ is the number of points in $\mathcal{P}$.
\end{lem}

For any set of segments $S$ and box $b$ and for each segment $s$, let $s_b$ be the part of $s$ that lies inside $b$.
We define $\ins(S, b) := \{ s_b \mid s \in S \text{ and } s \text{ has at least one endpoint in }b\}$ and
$\cross(S,b) := \{s_b \mid s \in S \text{ and } s \text{ has no endpoint in }b \}$.
Moreover we define $\out(S,b) := \{ s_{b'} \mid s \in S \text{ and } b \neq b'\}$.
Additionally, let $S(b) = \sum_{s \in S } \len(s_b)$.

We can now prove the two following structural Lemmas. See Fig. \ref{fig:lem_locality_TSP}  for an illustration of the proof.
\begin{figure}
\begin{center}
  \includegraphics{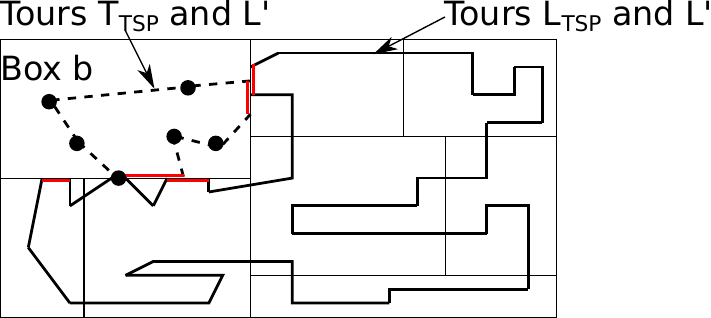}
  \caption{The solid black segments form the tour $L_\tsp$ outside $b$. The dotted line segments are the tour $T_\tsp$ inside $b$.
  The red segments are the one needed to connect the two tours.}
  \label{fig:lem_locality_TSP}
\end{center}
\end{figure}

\begin{lem}\label{lem:local_opt_st}
  Let $L_{\steiner}$ be a locally optimal solution to the Steiner Tree problem and let $T_{\steiner}$ be any Steiner Tree.
  Let $\mathcal{B}$ be a set of boxes produced by a dissection of $\mathcal{P} \cup L_\steiner \cup T_\steiner$.
  Using the same notation for a set of segments and their total length, we then have for any box $b \in \mathcal{B}$
  $$(1-\mathcal{O}(\eps^2)) L_\steiner(b) \le \ins(T_\steiner,b) + |\partial b| + L_\steiner/n,$$
  where $|\partial b|$ is the perimeter of $b$.
\end{lem}

\begin{proof}
  For each box $b$, the segments of $\cross(L_\steiner,b)$ can be distributed into
  6 different classes according to which side of $b$ they intersect.
  
  We divide further. Since the segments of a class are pairwise disjoint,
  there is a natural ordering of the segments inside each class.
  For each class that contains more than $1/\eps^2$ segments, we partition them into
  subsets that contain $\Theta(1/\eps^2)$ consecutive segments
  (in the natural order of the class).
  We define a sub-box for each subset of each class as follows. Let $s$ and $s'$ be the two extreme segments of the set in the ordering of the class.
  The sides of the sub-box associated to
  this subset consists of $s$ and $s'$ and the two shortest paths $p,p'$ along the sides of $b$ that connects the endpoints
  of $s$ and $s'$.

  Remark that the sum of the lengths of the sides of all the sub-boxes is at most $|\partial b| + \mathcal{O}(\eps^2 L_\steiner(b))$.  
  For each sub-box $b_0$, let $L'$ be the set of vertices of $L_\steiner$ that are outside $b_0$, plus the set of vertices of $T_\steiner$ that are
  inside $b_0$,
  plus the set of the intersection points of the edges of $L_\steiner$ and $T_\steiner$  with the sides of $b_0$.
  Thus, $L' \le \out(L_\steiner, b_0) + \ins(T_\steiner,b_0) + |\partial b_0|$.
  Moreover, we have $|L_\steiner \bigtriangleup L'| = \mathcal{O}(1/\eps^2)$ and the local near-optimality argument
  applies.
  Namely, we obtain that $(1-1/n) L_\steiner \le L'$, and so
  $$ - 1/n \cdot L_\steiner + \ins(L_\steiner,b_0) + \cross(L_\steiner, b_0) \le \ins(T_\steiner,b_0) + |\partial b_0|.$$
  We now sum over all sub-boxes of box $b$ and we obtain
  $$ L_\steiner(b) = \ins(L_\steiner,b_0) + \cross(L_\steiner, b_0) \le \ins(T_\steiner,b) + |\partial b| + \mathcal{O}(\eps^2 L_\steiner(b)) + L_\steiner/n.$$
\end{proof}

\begin{lem}\label{lem:local_opt_tsp}
  Let $L_{\tsp}$ be a locally optimal solution to the Traveling Salesman problem and let $T_{\tsp}$ be any tour. 
  Let $\mathcal{B}$ be a set of boxes produced by a dissection of $\mathcal{P}$.
  Using the same notation for a set of segments and their total length, we then have for any box $b \in \mathcal{B}$
  $$(1-\mathcal{O}(\eps^2)) L_\tsp(b) \le \ins(T_\tsp,b) + 3|\partial b|/2 + L_\tsp/n,$$
  where $|\partial b|$ is the perimeter of $b$.
\end{lem}
\begin{proof}
  We again further divide the boxes into sub-boxes as we did for Lemma \ref{lem:local_opt_st}. 
  For each sub-box $b_0$, we define a tour $L'$ obtained by a traversal of the following Eulerian graph.
  The graph vertices are $\mathcal{P}$, plus the corners of $\partial b_0$, plus all points of intersection of $L_\tsp$ and $T_\tsp$
  with $\partial b_0$.
  The edges are the segments of $\out(L_\tsp,b_0)$, plus the segments of $\ins(T_\tsp, b_0)$, plus $\partial b_0$ (so that the result is connected),
  plus a minimum length matching of the odd vertices of $\partial b_0$ (so that the result is Eulerian).
  Thus, $L' \le \out(L_\tsp,b_0) + \ins(T_\tsp,b_0) + 3|\partial b_0|/2$.
  
  Since the number of edges of $L$ intersecting $b_0$ is $\mathcal{O}(1/\eps^2)$ and the number of edges in $\ins(T_\tsp,b_0)$
  is $\mathcal{O}(1/\eps^2)$, we have $|L_\tsp \bigtriangleup L'| = \mathcal{O}(1/\eps^2)$ and the local near-optimality argument applies.  
  Namely, we obtain $(1-1/n) L_\tsp \le L'$, and so
  $$ - 1/n \cdot L_\tsp + \ins(L_\tsp,b_0) + \cross(L_\tsp,b_0) \le \ins(T_\tsp,b_0) + 3|\partial b_0|/2.$$
  We now sum over all sub-boxes of box $b$ and we obtain
  $$ L_\tsp(b) = \ins(L_\tsp,b) + \cross(L_\tsp,b)   \le \ins(T_\tsp,b) + 3|\partial b|/2 + \mathcal{O}(\eps^2 L_\tsp(b)) + L_\tsp/n.$$
\end{proof}


We can now prove Theorem \ref{thm:ub_tsp_st}.


\begin{proof}[Proof of Theorem \ref{thm:ub_tsp_st}]
  We first consider the Traveling Salesman case. Let $L_\tsp$ be a tour produced by Algorithm \ref{algo:LS} and $T_\tsp$ be any tour.
  Lemma \ref{lem:local_opt_tsp} implies that for any box $b$, we have
  $$(1-\mathcal{O}(\eps^2)) L_\tsp(b) \le \ins(T_\tsp,b) + 3|\partial b|/2 + L_\tsp/n.$$
  Since there are $\mathcal{O}(\eps^2 n)$ boxes in total, by summing over all boxes, we obtain
  $$-\mathcal{O}(\eps^2 L_\tsp) + \sum\limits_{b \in \mathcal{B}} L_\tsp(b) = (1 - \mathcal{O}(\eps^2)) L_\tsp \le \sum\limits_{b \in \mathcal{B}} (\ins(T_\tsp,b) + 3|\partial b|/2) \le T_\tsp + \frac{3}{2} \sum_{b \in \mathcal{B}} |\partial b|.$$
  By Lemma \ref{lem:Karp77}, $\sum_{b \in \mathcal{B}} |\partial b| = \mathcal{O}(\eps \sqrt{n})$  and so,
  $$(1-\mathcal{O}(\eps^2)) \cdot L_\tsp \le T_\tsp + \mathcal{O}(\eps \sqrt{ n}).$$
  
  To prove the Steiner Tree case, it is sufficient to notice that the total number of vertices in $\mathcal{P} \cup L_\steiner \cup T_\steiner$ is
  at most $3n$.
  It follows that the total number of boxes is $\mathcal{O}(\eps^2 n)$ and by Lemma \ref{lem:Karp77},
  $\sum_{b \in \mathcal{B}} |\partial b| = \mathcal{O}(\eps \sqrt{n})$.
  We apply a reasoning similar to the one for the TSP case to conclude the proof.

\end{proof}

Notice that we do not assume that the points are randomly distributed in the $[0,1]^2$ for the proofs of
Lemmas \ref{lem:local_opt_st} and \ref{lem:local_opt_tsp} and Theorem \ref{thm:ub_tsp_st}, thus they hold in the worst-case.

\begin{rem}
  One can ask whether it is possible to prove that the local search for TSP is a PTAS without the random input assumption.
  However, as shown in Fig. \ref{fig:lb_tsp} there exists a set of points such that there is a local optimum whose length
  is at least $(2-o(\eps)) \cost(\opt)$.
\end{rem}

\section{Clustering Problems}

  




We now tackle the analysis of the local search algorithm for some Clustering problems.
Recall that $L$ and $G$ denote the local and global optima respectively.
In the following, for each facility $l$ of $L$ (resp. $G$), we denote by
$V_L(l)$ (resp. $V_G(l)$) the Voronoi cell of $l$ 
in the Voronoi diagram induced by $L$ (resp. $G$).
We extend this notation to any subset $F$ of $L$, namely,
$V_L(F)$ denotes the union of the Voronoi cells of the facilities
of $F$ induced by $L$.
We define a recursive randomized decomposition (Algorithm \ref{algo:dissection}) based on $L $ and $ G$ (and the Voronoi cells induced by $L$).
This decomposition produces a tree encoded by the function Children(), where each node is associated to a region of the Euclidean plane.
In the first step of the dissection,  $B$ is the smallest square that contains all the facilities 
of $L \cup G$. At every recursive call of the procedure for $(B_r,L_r,G_r)$, the algorithm  maintains the following invariants:
\begin{itemize}
\item $B_r$ is a rectangle of bounded aspect ratio;
\item $L_r$ consists of all the facilities of $L$ that are contained in $B_r$;
\item $G_r$ consists of all the facilities of $G$ that are contained in $B_r$, plus some
  facilities of $G$ that belong to $V_L(L_r)$.
\end{itemize}

\begin{algorithm}[H]

\caption{Recursive Adaptive Dissection Algorithm}
\label{algo:dissection}
\begin{algorithmic}[1]
\Procedure{Adaptive\_Dissection}{$B, L, G, V_L$}
  \If{$|L| + |G| \ge 1/2\eps^2 $ }
      \If{$|L| > 1/2\eps$}
          \State \underline{Sub-Rectangle Process}:
          \State $B'\gets $ minimal rectangle containing all facilities of 
          $L$ in $B$
          \State $b' \gets$ maximum side-length of $B'$
          \State $B'_+ \gets$ Rectangle centered on $B'$ and extended by $b'/3$ in all four directions.
          \State $B''\gets B'_+\cap B$
          
          \State \underline{Cut-Rectangle Process}:
          \State $s''\gets $maximum side-length of $B''$
          \State $\ell \gets $line segment that is orthogonal to the side of length $s''$ and intersects it in a random position in the middle $s''/3$.
          \State Cut $B''$ into two rectangles $B_1$ and $B_2$ with $\ell$.
          \State

          \State Children($B$) $\gets \{B_1,B_2\}$
          \State $L_1 \gets L \cap B_1$ 
          \State $L_2 \gets L \cap B_2$
          \State $G_1 \gets G \cap \{g \mid g \in V_L(L_1) \text{ and } g \notin B_2\}$
          \State $G_2 \gets G \setminus G_1$
          \State \textsc{Dissection}($B_1$, $L_1$, $G_{1}$, $V_L$)
          \State \textsc{Dissection}($B_2$, $L_2$, $G_{2}$, $V_L$)
      \Else
          \State \underline{Partition Process}:
          \State Children($B$) $\gets$ Arbitrary partition of the facilities of $L \cup G$ in parts of size in $[1/2\eps^2, 1/\eps^2]$
    \EndIf  
\EndIf

\EndProcedure
\end{algorithmic}
\end{algorithm}

\textbf{Regions.} We now introduce the crucial definition of \emph{regions} of a dissection tree $\mathcal{T}$ of solutions $L$ and $G$.
For any node $N$ of the dissection produced by the Partition Process,
we consider that the associated rectangle is the bounding box of the
facilities of $L_N \cup G_N$.
We assign labels to the nodes of the tree. The label of a leaf $B$ 
is $|L_B| + |G_B|$. Then we proceed bottom-up, for each node of 
the tree, the labels of a node is equal to the sum of the labels of its two children.
Once a node has a label greater than $1/2\eps^2$, we say that this node is a \emph{region node} of the tree
and set its label to 0.
We define the regions according to the region nodes. For each region node $R$, the associated region is the rectangle defined by the 
node minus the regions of its descendants, namely minus the rectangles
of nodes of label 0 that are descendants of $R$. See Fig. \ref{fig:regions} for an illustration of the regions.
In the following, we denote by $\mathcal{R}$ the set of regions.


\textbf{Portals.} Let $\mathcal{D}$ be a dissection produced by 
Algorithm \ref{algo:dissection}.
For any region $R$ of $\mathcal{D}$ not produced by the Partition Process, 
we place $p$ equally-spaced \textit{portals} along each boundary of 
$R$.
We  refer to the dissection $\mathcal{D}$ along with the associated
portals as $\mathcal{D}_p$.
See Fig. \ref{fig:regions} for more details on the regions and portals.


\begin{figure}
  \centering
  \includegraphics[scale=0.3]{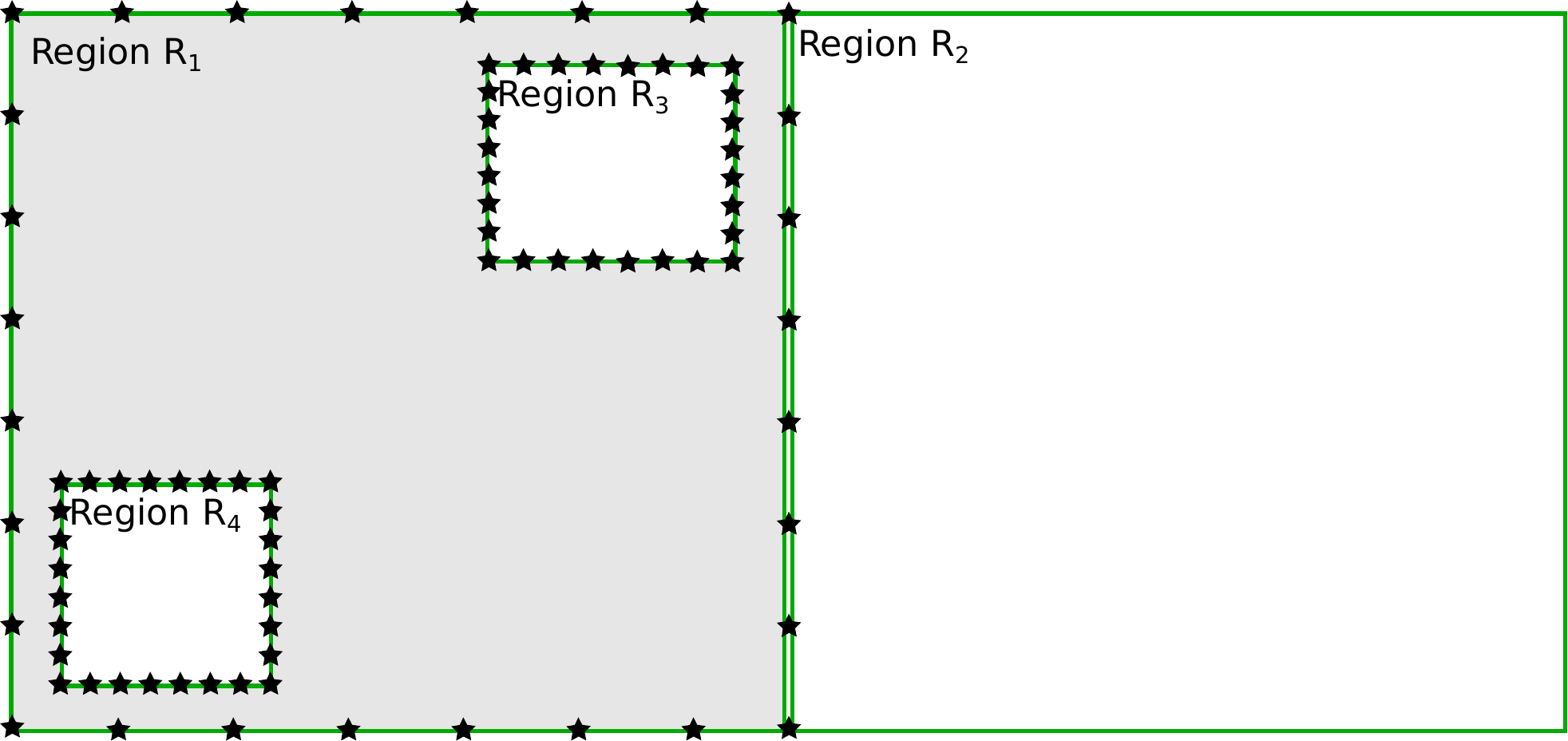}
  \caption{Details of the regions and portals associated to a dissection. The star-shaped points are the portals associated to Region $R_1$.
  Regions $R_2,R_3,R_4$ are the only regions sharing portals with region $R_1$. All the regions are disjoint.}
  \label{fig:regions}
\end{figure}

\textbf{Definitions and Notations.}
For any clustering problem, we denote by $\mathcal{C}$ the sets of the input points. 
We refer to an input point as a \textit{client}.
A solution to a clustering problem is a set of facilities $S \subset \mathds{R}^2$.

For any solution $S$ and any client $c$, we denote by $c_S$ the distance from client $c$
to the closest facility of $S$: $c_S = \min\limits_{s \in S} d(c,s)$.
The service cost of a solution $S$ to a clustering problem is $\sum\limits_{c \in \mathcal{C}} c_S$.
Additionally, for any solution $S$ and client $c$, we define $c(S)$ as the facility of $S$ that
serves $c$ in solution $S$, namely $c(S) := \argmin_{s \in S} d(c,s)$

Let $B$ be the smallest rectangle that contains all the clients.
Let $L$ and $G$ be two sets of facilities.
We now give the definition of an assignment which is crucial for the main proposition.
\begin{defn}
  We define an \textit{assignment} as a function that maps the clients to the facility of $L \cup G$.
\end{defn}
Let $E_0$ be the assignment that maps each client $c$ to the facility of $\{c(L),c(G)\}$ that is the farther,
namely, $\forall c \in \mathcal{C}$, $E_0(c)=\argmax (dist(c,c(G)),dist(c,c(L))).$

We show the following proposition which is the technical center of the proof.

\begin{prop}
  \label{thm:new_struct}
  Let $1/\eps^2 > 0$ be an integer, $G$ and $L$ be two sets of facilities.
  Let $\mathcal{D}_{1/\eps^2}$ be a dissection tree with portals.
  There exists an assignment $E$ that satisfies the following properties.
  Let $R$ be a region not produced by the Partition Process.
  If a client $c$ is such that $c(L) \in R$ and $c(G) \notin R$ then $E(c)$ is
  either a portal of $R$ or a facility of $L \setminus R$.\\
  Moreover, 
  $$\mathds{E} [\sum_{c \in \mathcal{C}} | dist(c,E(c))-dist(c,E_0(c))|]  =\sum\limits_{c \in \mathcal{C}} \mathcal{O}(\eps^2 \log(1/\eps^2) \cdot (c_G + c_L)).$$

\end{prop}

We start by proving some properties of Algorithm \ref{algo:dissection}\footnote{Lemma \ref{lem:sub-rect} is essentially
  Lemma 4 from \cite{Kolliopoulos07} but a careful writing of the details of the calculation reveals slightly different constants.}.
The proofs of the following Lemmas are deferred to Appendix \ref{appx:struct_thm}.

\begin{defn}[Aspect Ratio]
  We define the aspect ratio of a rectangle $R$ that has sides of lengths $r$ and $r'$ as $\max(\frac{r}{r'},\frac{r'}{r})$.
\end{defn}

\begin{lem}\label{lem:aspect_ratio}
  Let $R$ be a rectangle produced by either the Sub-Rectangle or the Cut-Rectangle process of Algorithm \ref{algo:dissection}.
  The aspect ratio of $R$ is at most $5$.
\end{lem}

\begin{lem}[\cite{Kolliopoulos07}]\label{lem:sub-rect}
Let $l \in L$ be a facility and $v \in \mathds{R}^2$ be any point. Let $d$ be the distance between $v$ and $l$. 
If a cutting line segment $s$ produced by the Sub-Rectangle process during Algorithm \ref{algo:dissection}
separates $v$ and $l$ for the first time, then $\len(s) \le 5d.$
\end{lem}

\begin{lem}\label{lem:4-generalized}
Let $L$ be a set of facilities.
Let $v \in \mathds{R}^2$, $l \in L$, $d_0=\dist(v,l)$.
Suppose that, in Algorithm \ref{algo:dissection}, $v$ and $l$ are first separated by a line $s$ that is vertical and that $l$ is to the right of $s$.
Let $d_1$ be the distance from $v$ to the closest open facility located to its left.
Then, the length of $s$ is either:
(i) larger than $d_1/4$ or 
(ii) smaller than $12d_0$.

\end{lem}


\begin{lem}[\cite{Kolliopoulos07}]\label{lem:fact1}
  Let Event$_0(d, s)$ denote the event that an edge $e$ of length $d$ is separated by a cutting line of
  side-length $s$ that is produced by Cut-Rectangle.\\
  Then, $Pr[\text{Event}_0(d, s)] \le 3d/s.$
\end{lem}


We now show the proof of the Structure Theorem.

\begin{proof}[Proof of Proposition \ref{thm:new_struct}]
  Let $p := 1/\eps^2$.
  By linearity of expectation, we only need to show  this on a per-client basis. 

  Let $c$ be a client and $R$ a region containing $l :=c(L)$ but not $g :=c(G) $. 
  Let $B$ be the first box of the dissection, in top-down order, that contains $l$ but not $g$, 
  and let $s$ be the side of $B$ that is crossed by $[l,g]$. 
  We have: $dist(g,l)\leq dist(g,c)+dist(c,l)=c_G+c_L$. 
  Up to a rotation of center $g$, $l$ is to the north-west of $g$.
  Let $u, w$ be the closest facilities of $L$ respectively to the south and to the east of $g$.

  To construct $E$, we start with $E := E_0$, and modify $E$ one client at a time so that each client satisfies the first property,  and we bound
  the corresponding expected cost increase. The initial cost of $E$  is  
  $\sum\limits_{c \in \mathcal{C}} \max(c_G,c_L)$.  
  We modify $E(c)$ depending on whether $s$ is vertical or horizontal and according to the length of $s$.
  We first provide an upper bound on the expected cost increase induced by $E(c)$ for the case where $s$ is vertical.
  It is easy to see that, when $s$ is horizontal, applying the same reasoning on $w$ instead of $u$ leads to 
  an identical cost increase and thus, the total cost increase is at most twice the cost increase computed 
  for the case where $s$ is vertical.

  By Lemma \ref{lem:4-generalized}, the following cases cover all possibilities for the case where $s$ is vertical.
  \begin{itemize}
  \item $s$ is vertical and $s$ was produced by Sub-Rectangle. Then we define $E(c)$ as the portal on $s$ that is closest to $[g,l]$. By Lemma \ref{lem:sub-rect}, the cost increase   is  at most $\mathcal{O}((c_G+c_L)/p)$.
  \item $s$ is vertical and $s$ was produced by Cut-Rectangle  and its length is at most $12(c_L+c_G)$. Then again we define $E(c)$ as the portal on $s$ that is closest to $[g,l]$. By assumption, again the cost increase is at most $\mathcal{O}((c_G+c_L)/p)$.
  \item $s$ is vertical and $s$  was produced by Cut-Rectangle  and its length is greater that $12(c_L+c_G)$.
    Lemma \ref{lem:4-generalized} implies that $s$ has length greater than $d_{u}/4$.
    If the length of $s$ is in $[d_u/4,p d_u]$. Then again we define $E(c)$ as the portal on $s$ that is closest to $[g,l]$.  
    Let $\mathcal{E}_0$ be the event that $d_u/4 \le |s| \le p \cdot d_u$ and $s$ is vertical.
    The expected cost increase in this case is, by Lemma \ref{lem:fact1}, at most 
    $$\sum\limits_{\substack{d_u/4 \le i \le p \cdot d_u \\ \text{s.t } i/d_u \text{ is power of 2}}} pr[|s| = i\text{ and } \mathcal{E}_0] \cdot (i/p) \le
    \mathcal{O}(\log(p)/p \cdot (c_G + c_L)).$$


  \item We now turn to the last case. Namely, $s$ was produced by Cut-Rectangle and its length is greater than or equal to $p \cdot d_u$.
    We define $E(c)$ depending on whether $u$ is in $R$ or not. This leads to two different sub-cases.

    1. $u \notin R$.                
        Then we define $E(c):=u$.
        The cost is bounded by the cost to go to $g$ ($\max(c_G,c_L)$) plus
        the cost to go from $g$ to $u$, which is $d_u$.
        Let $\mathcal{E}_1$ be the event that $u \notin R$ and $p \cdot d_u < |s|$ and $s$ is vertical. 
        The cost increase is, by Lemma \ref{lem:fact1}, at most,
        $$\sum\limits_{\substack{i > p \cdot d_u  \\ \text{s.t } i/d_u \text{ is power of 2}}} pr[|s| = i\text{ and } \mathcal{E}_1] \cdot (d_u)
        \le \mathcal{O}((c_G + c_L)/p).$$

    %
    2. $u \in R$. 
        Let $d$ denotes the first line that separates $u$ from $g$. Since $u$ is to the right of $g$,
        $d$ is different from $s$ and has size at least $d_u$. We have two sub-cases.
        
        First, if $d$ was produced before $s$ in the dissection, then we also have $|d| > |s|$.
        Let $\mathcal{E}_2$ be the event $|d| > |s| > p \cdot d_u$ and $s$ is vertical.
        We now fix $d$. We assign $E(c)$ to be the closest portal on $R$, the expected cost increase conditioned upon $d$ is then at most:
        $$\sum\limits_{\substack{p \cdot d_u < i \le |d| \\ \text{s.t } i/d_u \text{ is power of 2}}} pr[|s| = i\text{ and } \mathcal{E}_2] \cdot (i/p)
        \le \mathcal{O}(\log(\frac{|d|}{p \cdot d_u}) \cdot (c_G + c_L)/p).$$
        
        We then remove the conditioning on $d$. 
        If $d$ was produced by the Sub-Rectangle process, then $p \cdot d_u < |d| \le 5d_u$ by Lemma \ref{lem:sub-rect} 
        and the expected cost increase is at most $\mathcal{O}((c_G+c_L)/p)$.
        Otherwise, $d$ was produced by the Cut-Rectangle process, and then the expected cost increase is at most 
        $$\sum\limits_{\substack{i > p \cdot d_u \\ \text{s.t } i/d_u \text{ is power of 2}}} pr[|d| = i\text{ and } \mathcal{E}_2] \cdot \mathcal{O}(\log(\frac{i}{p\cdot d_u}) \cdot (c_G + c_L)/p)
        \le
        \mathcal{O}((c_G + c_L)/p).$$\\
        
        Second, if $d$ was produced after $s$ in the dissection, namely $|s| > |d|$.
        Let $\mathcal{E}_3$ denote the event that $|s| > |d|$ and $|s| > p \cdot d_u$ and $s$ is vertical.
        We assign $c$ to the closest portal located on $d$, which is at distance at most $d_u + |d|/p$ from $g$ (and so at distance at most 
        $c_G + d_u + |d|/p$ from $c$).
        We start by fixing $s$. The expected cost conditioned upon $s$ is then (no matter how $d$ was produced), at most\\
        $$\sum\limits_{\substack{d_u < i < |s| \\ \text{s.t } i/d_u \text{ is power of 2}}} pr[|d| = i\text{ and } \mathcal{E}_3] \cdot (d_u + i/p)$$
        We then remove the conditioning on $s$, which leads to an expected cost of at most
        $$\sum\limits_{\substack{j > p \cdot d_u \\ \text{s.t } i/d_u \text{ is power of 2}}} pr[|s| = j \text{ and } \mathcal{E}_3]  \sum\limits_{d_u < i < j} 3(d_u/i) \cdot (d_u + i/p) \le         \mathcal{O}((c_G + c_L)/p)$$
        


    Thus, the total expected cost increase for $E$ is at most $\mathcal{O}((\log(p)/p) \cdot (c_G + c_L))$.

  \end{itemize}
  \end{proof}

\textbf{Partitioning the Clients and the Facilities.}
Before going further, we need to define a partition of the clients and 
the facilities
according to the dissection produced by Algorithm \ref{algo:dissection}.

We partition the clients into two sets $C_G$ and $C_L$. $C_G$ contains the clients
that are closer to a facility of $G$ than to a facility of $L$
and $C_L$ contains the rest of the clients, namely $C_G := \{c \mid c_L = \max(c_L,c_G) \}$ 
and  $C_L := \{c \mid c_G \neq \min(c_L,c_G)\}$.
Let $\mathcal{D}$ be a dissection produced by Algorithm 
\ref{algo:dissection} and the set of its associated regions $\mathcal{R}$.
For any region $R$, we denote $C_{G}(R)$ the set of clients that are served 
by $G_R$ in $G$ and that do not lay on a region not in $P$.
Furthermore, we define $C_L(R)$ as the set of clients that are served by $L_R$ in $L$ 
and let $C_R := V_{G}(G_R) \setminus \left( C_L \cap (V_L({L \setminus L_R}) \right)$
\footnote{This can be rewritten as $C_R := V_G(G_R) \cap (C_G \cup V_L(L_R))$.}.   
This set contains the clients served by $G_R$ in $G$ except those 
that belong to $C_L$ and that are served by $L \setminus L_R$ in $L$.
See Fig. \ref{fig:NR} for an illustration.
Additionally, we define $\Delta_R := V_L(L_R) \setminus V_G(G_R)$.

\begin{figure}
  \centering
  \includegraphics[scale=0.3]{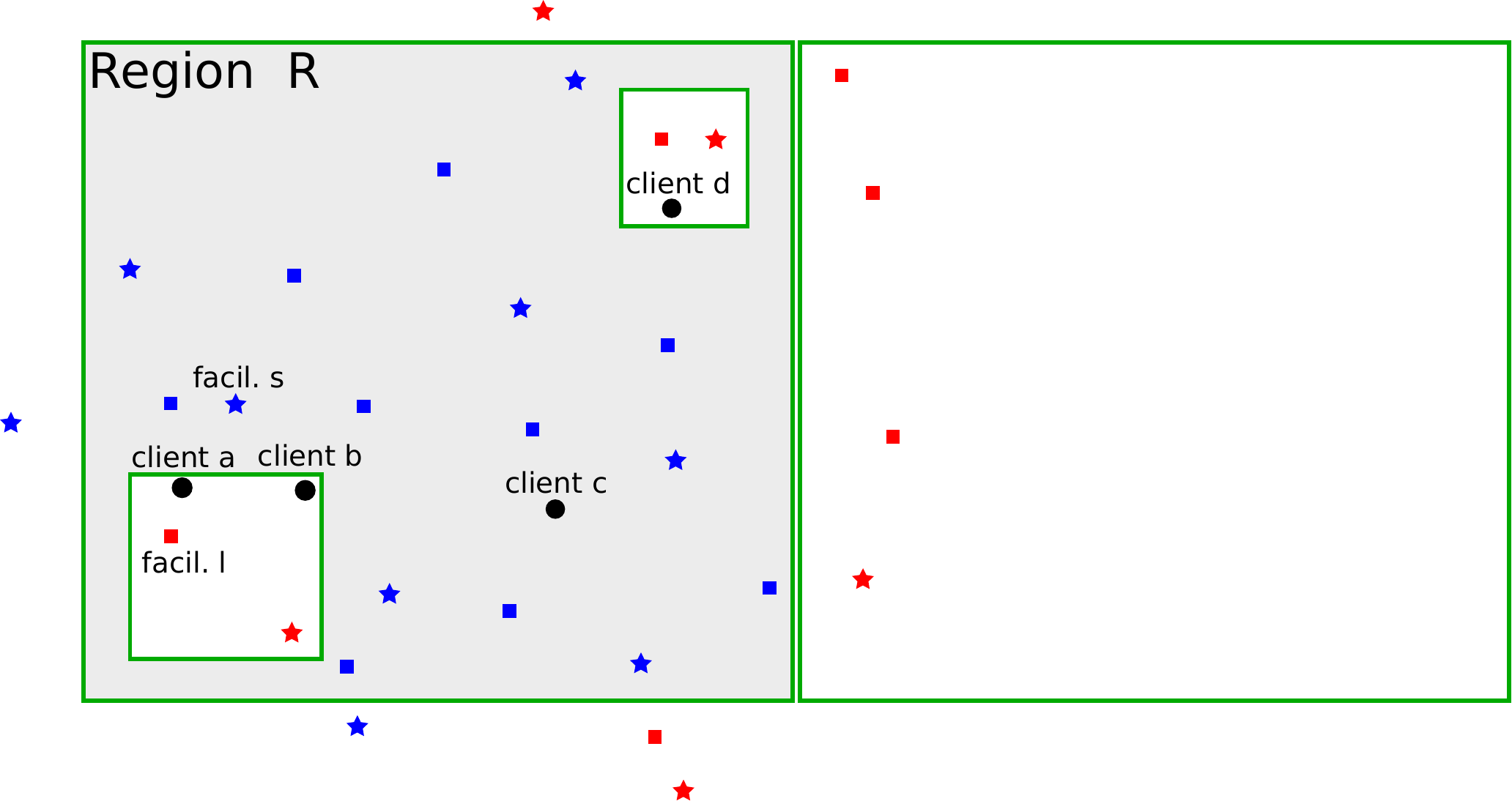}
  \caption{Details of the partitioning of the client.
    The star-shaped points are the facilities of $G$ and the square-shaped one are the facilities of $L$.
    The blue star-shaped and square-shaped belong to respectively $G_R$ and $L_R$.
    Since client $a$ is closer to facility $l$ than to facility $s$, it belongs to the set $C_L$.
    Moreover, it is served in $L$ by a facility that does not belong to $V_L(L_{R})$, and so, it is not included in set $C_{R}$. 
    Client $b$ is closer to facility $s$ than to facility $l$ and so, it is included in set $C_{R}$ albeit it is served by a facility located on 
    another region in $L$.
    Client $c$ is served by a facilities that belongs to $V_L(L_{R})$  (in $L$ and $G$) and so, it belongs to $C_{R}$. Finally, client $d$ 
    does not belong to $V_{G}(G_R)$ and so, is no included in set $C_{R}$.}
  \label{fig:NR}
\end{figure}

\subsection{Facility Location}

We now prove the approximation ratio of Algorithm \ref{algo:LS} for
facility location.

\begin{thm}
  \label{thm:FL}
  For Facility Location, Algorithm \ref{algo:LS} produces a solution $L$ of cost at most $(1 + \mathcal{O}(\eps)) \cdot \cost(\opt)$.
\end{thm}

\begin{proof}
  Let $\opt$ be a globally optimum solution and $L$ be a locally optimum solution.
  By Proposition \ref{thm:new_struct}, for any $p>0$ there exists an assignment $E$ for each random dissection $\mathcal{D}_p$ with portals 
  of $L \cup \opt$,
  such that for any client $c$ and region $R$, if $c(L) \in R$ and $c(G) \notin R$ then 
  $c$ is served by a portal of $R$ or a facility of $L \setminus R$ in $E$ and
  the expected cost of $E$ is at most $\mathds{E} = \sum\limits_{c \in \mathcal{C}} \max(c_L,c_G) +  \mathcal{O}(\log(p)/p \cdot (\sum\limits_{c \in \mathcal{C}} (c_G + c_L)))$.
  This implies that there exists a dissection $\mathcal{D}_p$ for which $E$ has value at most $\mathds{E}$.

  Throughout the proof, we  consider this dissection $\mathcal{D}_p$ and fix $\eps := \log(p)/p$.
  Let $\mathcal{R}$ be the set of regions associated to $\mathcal{D}_{p}$.
  We start by constructing a solution $G$ based on $\opt$ and we  compare the cost of $L$ to the cost of $G$.
  The solution $G$ contains all the facilities of $\opt$ plus some extra facilities.
  First, it has one facility at each portal of $\mathcal{D}_{p}$.
  Moreover, for each region $R$ that is produced by the Partition Process, 
  we open the facilities of $L_R$. Recall that for each of these regions, $|L_R| \le 1/\eps $.
  We keep the same assignment for the clients.
  Since there are $\mathcal{O}(\eps^2 (|G|+|L|))$ regions and that for each region $G$ uses at most $1/\eps$ extra facilities, 
  the cost of $G$ is at most Cost($\opt$)+ $\mathcal{O}(\eps (|\opt| + |L|)f)$.
  We now prove that the cost of $L$ is at most $(1 + \mathcal{O}(\eps))/(1 - \mathcal{O}(\eps))$ times the cost of $G$, namely
  $$
  |L| \cdot f + \sum\limits_{c \in \mathcal{C}} c_L \le (\frac{1 + \mathcal{O}(\eps)}{1 - \mathcal{O}(\eps)}) (|G|\cdot f 
  + \sum\limits_{c \in \mathcal{C}} c_G).$$


  We focus on the cost of a region $R$. 
  We show that, by local optimality, for each region $R$, replacing solution $L$ by solution $G$ does not lead to a much better cost.
  We serve the clients of $C_R$ optimally (namely by the facilities that serve them in $G$) and
  the clients of $L_{R} \setminus G_R$ by the facilities located on the portals of $R$ or by the facilities
  of $L \setminus L_R$, depending on whether they belong to $C_L$ or $C_G$ and according to the assignment $E$.  
  Since $|L_R \setminus G_R| + |G_R \setminus L_R|  = \mathcal{O}(\eps^{-3})$, the locality argument applies.
  Namely, we have
  $$
  (|G_R| - |L_R|)f + \sum\limits_{c \notin C_R \cup \Delta_R}c_L + \sum\limits_{c \in C_R} c_G + \sum\limits_{c \in \Delta_R} c_E 
  \ge (1-1/n)  (|L|f + \sum\limits_{c} c_L).
  $$
  The rest of the proof is mainly computational and can be found in 
  the appendix \ref{appx:FL}.

\end{proof}

\subsection{$k$-Median}
Let $L$ and $\opt$ be respectively local and global optimal solutions to
the $k$-Median problem.
We start with a technical Lemma which allows us to find
``clusters'' of regions of the plane that have roughly the same number of facilities
of $L$ and $G$. See Fig. \ref{fig:balanced_cluster} for an illustration.
The proof of the Lemma is deferred to Appendix \ref{appx:kmed}.

\begin{lem}[Balanced Clustering]
  \label{lem:kmed_partition}
  Let $\mathcal{R} = \{r_1,...,r_p\}$ be a collection of disjoint sets. Each set contains elements of type either $L$ or $G$
  and has size at least $1/2\eps^2$ and at most $1/\eps^2$. The total number of elements of type $L$
  is $(1+3\eps)$ times higher than the number of elements of type $G$.

  There exists a clustering of $\{r_1,...,r_p\}$ in clusters
  satisfying the following two properties. For any cluster $C$, 
  \begin{itemize}
  \item $C$ contains at most $\mathcal{O}(1/\eps^5)$ elements of $\mathcal{R}$, namely $|C| = \mathcal{O}(1/\eps^5)$;
  \item the difference between the number of elements of $L$ in the sets contained in $C$ 
    and the number of elements of $G$ in the sets contained in $C$ is at least $|C|/\eps$:
    $$\sum_{r_i\in C} |r_i \cap L| - \sum_{r_i\in C} |r_i \cap G| \geq |C|/\eps,$$
    for any $1/\eps \in \mathds{N}$.
  \end{itemize}


\end{lem}

\begin{figure}
  \begin{center}
    \includegraphics[scale=0.25]{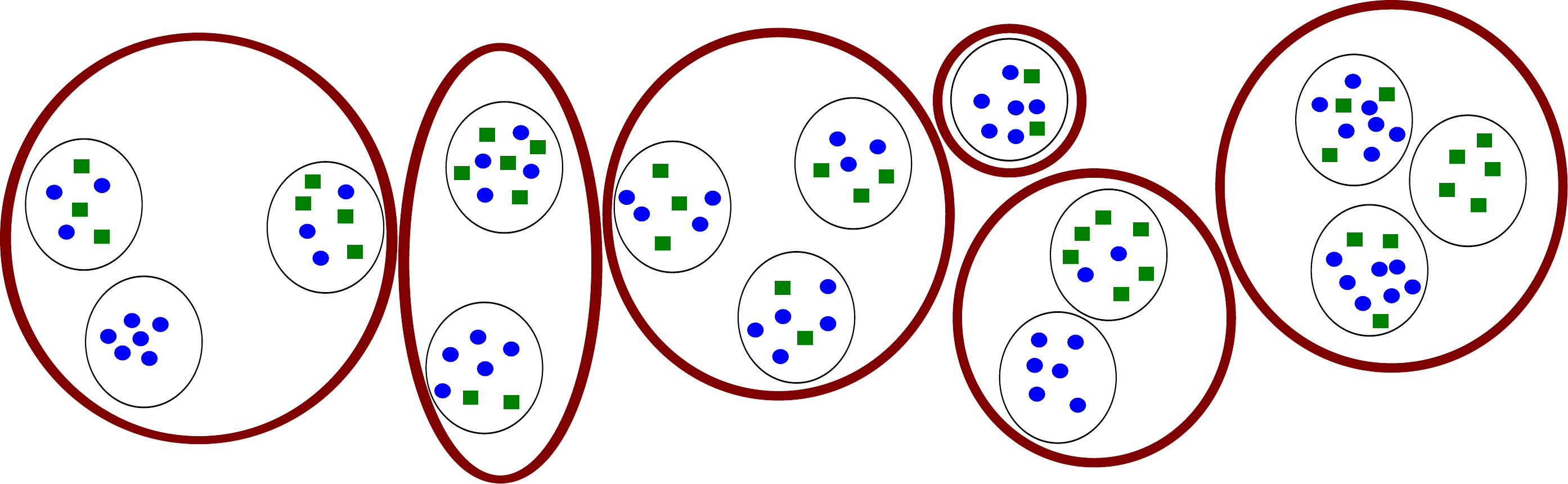}
    \caption{The circle-shaped points are the elements of type $L$ and the square-shaped ones the elements of type $G$.
      The black circles mark the sets $\{r_1,\ldots,r_p \}$ and the red ones show a clustering of those sets
      that satisfy the property of Lemma \ref{lem:kmed_partition}.
    }
    \label{fig:balanced_cluster}
  \end{center}
\end{figure}

\begin{thm}
  \label{thm:kmed}
  For $k$-Median, Algorithm \ref{algo:LS} for $k$-Median produces a solution $L$ of cost at most $(1 + \mathcal{O}(\eps) \cost(\opt)$ using at most $1 +
  \mathcal{O}(\eps))k$ Medians. 
\end{thm}

\begin{proof}
  Remark first that solution $L$ uses $(1+\mathcal{O}(\eps))k$ facilities.
  We now show that the cost of solution $L$ is at most 
  $1 + \mathcal{O}(\eps)$ times higher than the cost of the optimal solution.\\
  Recall that by Proposition \ref{thm:new_struct}, for any $p>0$ there exists an assignment $E$ for each random dissection $\mathcal{D}_p$ 
  of $L \cup \opt$ with portals, such that for any client $c$ and region $R$, if $c(L) \in R$ and $c(\opt) \notin R$ then 
  $c$ is served by a portal of $R$ or a facility of $L \setminus R$ in $E$ and
  the expected cost of $E$ is at most $\mathds{E} = \sum\limits_{c \in \mathcal{C}} \max(c_L,c_\opt) +  \mathcal{O}(\log(p)/p \cdot (\sum\limits_{c \in \mathcal{C}} (c_\opt + c_L)))$. \\
  This implies that there exists a dissection $\mathcal{D}_p$ for which $E$ has value at most $\mathds{E}$.
  Throughout the proof, we  consider such a dissection $\mathcal{D}_p$ and fix $\eps := \log(p)/p$.
  Let $\mathcal{R}$ be the set of regions associated to $\mathcal{D}_{p}$.
  We prove that the cost of $L$ is at most $(1 + \mathcal{O}(\eps))/(1 - \mathcal{O}(\eps))$ times the cost of $S$, namely
  $$
  \sum\limits_{c \in \mathcal{C}} c_L \le \frac{1 + \mathcal{O}(\eps)}{1 - \mathcal{O}(\eps)} \sum\limits_{c \in \mathcal{C}} c_\opt.
  $$


  Let $\mathcal{P}$ be a clustering of the regions satisfying the properties of Lemma \ref{lem:kmed_partition} 
  (depending on $L$ and $\opt$).
  We start by constructing a solution $G$ based on $\opt$ and we  compare the cost of $L$ to the cost of $G$.
  We construct $G$ in a similar way to in the proof of Theorem \ref{thm:FL}.
  Namely, the solution $G$ contains all the facilities of $\opt$ plus some extra facilities:
  one facility at each portal of $\mathcal{D}_{p}$ and for each region $R$ that is produced by the Partition Process, 
  we open the facilities of $L_R$. Recall that for each of these regions, $|L_R| \le 1/\eps $.
  We keep the same assignment for the clients.
  We now compare the costs of $L$ and $G$. To do so, we  consider all the regions of each cluster of the clustering $\mathcal{P}$
  at the same time. Namely for each cluster $R$, $L$ uses at least as many facilities as $G$.
  Therefore $|S_P \setminus L| + |L \setminus S_P| = \mathcal{O}(1/\eps^9)$ and the locality argument applies.
  The rest of the proof is similar to the proof of \ref{thm:FL} and is mainly computational and can be found in Appendix \ref{appx:kmed}.

\end{proof}

\textbf{Higher Dimensions.}
Previous results generalize to any dimension $d$. It leads to  
algorithms that have exponential dependency in $d$. 
More precisely, for any dimension $d$, more portals are needed to maintain the expected cost increase
for the assignment $E$ provided by the Structure Theorem.
Each of the $2d$ faces of each region has to count $p^{d-1}$ portals.
Proposition \ref{thm:new_struct} generalizes to any dimension 
$d$ with $\mathcal{O}(d p^{d-1})$ portals instead of $p$.
For Facility Location, Condition($S',\eps$) has to be adapted to $|S \setminus S| + |S \setminus S| = \mathcal{O}(d/ \eps^{d+1})$.
Thus, Theorem \ref{thm:FL} still applies to show that the adapted Algorithm provides a $(1+\mathcal{O}(\eps))$ approximation. 
For the $k$-Median problem, Condition($S',\eps$) has to be adapted to $|S'| \le (1+3\eps)k$ and 
$|S \setminus S'| + |S' \setminus S| = \mathcal{O}(d/\eps^{7+d})$.
Theorem \ref{thm:kmed} still applies to prove the approximation ratio of the adapted Algorithm.





\bibliographystyle{plain}
\bibliography{euclideanFL}

\newpage

\appendix
\section{The Traveling Salesman Problem}
\label{appx:fig_tsp}
\begin{figure}[H]
  \centering
  \includegraphics[scale=0.3]{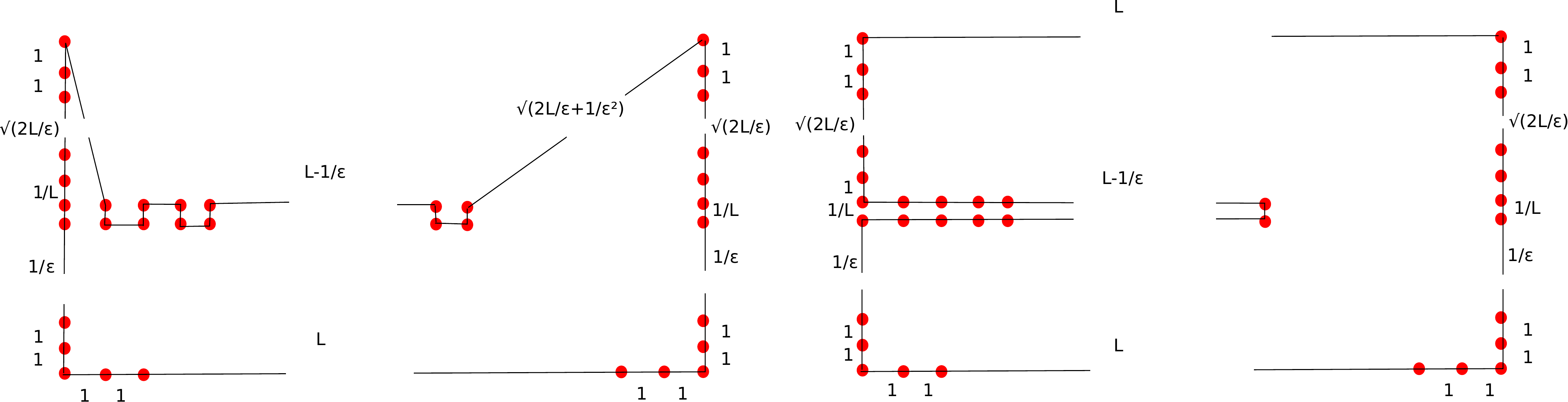}
  \caption{The tour on the right is k-optimal for any $k = o(\sqrt{L})$ but is $(2 - \mathcal{O}(1/k))$ times longer 
    than the tour on the left.}
  \label{fig:lb_tsp}
\end{figure}

\section{Polynomial-Time Local Search Algorithms}
\label{appx:ptime}

\begin{lemappx}[\ref{lem:polytime}]
  The number of iterations of Algorithm \ref{algo:LS} is polynomial for the Facility Location, the $k$-Median, the Traveling Salesman and the Steiner Tree Problems.
\end{lemappx}
\begin{proof}
  Let $\cost(L)$ denote the cost of a locally optimum solution and $\cost(S_0)$ denotes the cost of the initial solution, 
  then the number of steps in the algorithm is at most
  $$\frac{\log(\cost(S_0) / \cost(L))}{\log(\frac{1}{1-1/n})} .$$
  Since the cost of any minimal solution $S_0$ is at most $\mathcal{O}(n)$ (up to re-scaling the distances) and
  as $\log(n)$ and $\log(\cost(L))$ are polynomial in the input size,
  the algorithm terminates after polynomially many local search steps which are executed in polynomial time.
\end{proof}

\section{The Structure Theorem}
\label{appx:struct_thm}

\begin{lemappx}[\ref{lem:aspect_ratio}]
  Let $R$ be a rectangle produced by either the Sub-Rectangle or the Cut-Rectangle process of Algorithm \ref{algo:dissection}.
  The aspect ratio of $R$ is at most $5$.
\end{lemappx}
\begin{proof}
  We show by induction that at each step of the Dissection algorithm, the rectangle produced by the Sub-Rectangle or the Cut-Rectangle
  process has an aspect ratio of at most $5$.
  It is true at the first step of the algorithm since we consider the smallest square that contains all the facilities.
  We suppose that it is true up to depth $i$ of the Dissection algorithm and let $B$ be the input rectangle for step $i+1$. 
  Consider the Sub-Rectangle process applied to $B$.
  Let $B'$ be the smallest rectangle that contains the facilities of $B$ and $b'$ be the larger side of $B'$ and
  $s'$ be the smaller one.
  Let $r$ be the length of the side of $B$ that is parallel to $b'$ and $r'$ be the side of $B$ that is parallel to $s'$.
  Let $B''$ be the rectangle produced by the Sub-Rectangle process. The ratio of its two sides is at most 
  $$\max(\frac {\min(r, b' + 2b'/3)} {\min(r', s' + b'/3)},\frac{\min(r', s' + 2b'/3)}{\min(r, b' + b'/3)})$$
  We use the fact that
  $$\max(\frac{\min(a,b)}{\min(c,d)}) \le \max(\min(\frac{a}{c}, \frac{b}{d})).$$
  Thus, the ratio is at most
  $$  \max(\min\frac{(b' + 2b'/3)}{(s' + b'/3)}, \frac{r}{r'}, \frac{r'}{r}) \le 5$$
  since $1/5 \le r/r' \le 5$ by induction hypothesis and $(5b'/3)/(s'+b'/3) \le 5$.\\

  We now show that a rectangle produced by the Cut-Rectangle process has an aspect ratio of at most $5$.
  The Cut-Rectangle process on a rectangle that has side of length $s$ and $s'$, with $s \ge s'$, produces two rectangles that have 
  a side of length in $[s/3, 2s/3]$. The other side has length $s'$. 
  Thus, the aspect ratio of the new rectangles is at most $\max(\frac{2s/3}{s'},\frac{s'}{s/3})$.
  By induction hypothesis, we have $s/s' \le 5$ and so, since $s' \le s$, the aspect ratios of the new rectangles are at most $5$.

\end{proof}

\begin{lemappx}[\ref{lem:sub-rect}]
Let $l \in L$ be a facility and $v \in \mathds{R}^2$ be any point. Let $d$ be the distance between $v$ and $l$. 
If a cutting line segment $s$ produced by the Sub-Rectangle process during Algorithm \ref{algo:dissection}
separates $v$ and $l$ for the first time, then $\len(s) \le 5d.$
\end{lemappx}
\begin{proof}
  Let $s$ be the first line that separates $l$ from $v$ in the dissection. 
  Let $B$ be the last rectangle that contained both $v$ and $l$.
  Let $B'$ be the minimal rectangle that contains all the facilities of $L \cap B$ 
  and let $B''$ be the square centered on $B'$ produced by Algorithm 
  \ref{algo:dissection}.
  $B''$ is thus the square that produced $s$. 
  Since the Sub-Rectangle process focus on the intersection between $B''$ and $B$, the length of $s$ is at 
  most the side-length of $B''$. 

  Since $v$ is not in $B'$ and there is no facility in the outer fifths
  of the rectangle, $l$ is thus located on the middle part of $B''$ 
  and it follows that $\len(s)/5 \le d$.

  


  
\end{proof}

\begin{lemappx}[\ref{lem:4-generalized}]
Let $L$ be a set of facilities.
Let $v \in \mathds{R}^2$, $l \in L$.
Suppose that, in Algorithm \ref{algo:dissection}, $v$ and $l$ are first separated by a line $s$ that is vertical and that $l$ is to the right of $s$.
Let $d_0=\dist(v,l)$ and $d_1$ be the distance from $v$ to the closest facility of $L$ located to its left.
Then, the length of $s$ is either:
(i) larger than $d_1/4$ or 
(ii) smaller than $12 d_0$.
\end{lemappx}

\begin{proof}
  By a slight abuse of notation, $s$ denotes both the first line that separates $v$ from $l$ and its length. 
  Assume $s < d_1/4$. Either $s$ was produced by the Sub-Rectangle process, or by the Cut-Rectangle process.
  If $s$ was produced by the Sub-Rectangle process then by Lemma \ref{lem:sub-rect} $s$ has length at most $5d_0$.\\
  Otherwise $s$ was produced by the Cut-Rectangle process.
  We consider the last rectangle $R$ (in the top-down order) that 
  contains both $v$ and $l$ and let $r$ be the length of the 
  longer side of $R$.
  Observe that the diagonal of the rectangle to the left of $s$ is at most $\sqrt{s^2 + (2r/3)^2}$.
  By Lemma \ref{lem:aspect_ratio}, the aspect ratio of $R$ is at most $5$ and so, $r \le 5s$.
  Thus, the diagonal is at most $s \sqrt{1 + (2r/3)^2} < 4s$. Since $4s < d_1$, the part of $R$ to the left of $v$
  contains no facility.\\
  We write
  \begin{equation}
    \label{eq:lem4_1}
    d_0 \ge \dist(l, \text{left}(R)) - \dist(v, \text{left}(R)),
  \end{equation}  
  where $\dist(l, \text{left}(R)), \dist(v, \text{left}(R))$ are respectively the distances from $l$ and $v$ to the left side of $R$.\\
  Since $l$ is to the right of $s$,
  \begin{equation}
    \label{eq:lem4_2}
    \dist(l, \text{left}(R)) \ge r/3.
  \end{equation}
  
  Since $R$ is produced by the Sub-rectangle process, there is a parent $B$.
  We consider $B'$ as defined in Algorithm \ref{algo:dissection}.
  By definition $B'$ does not extend beyond the facilities of $R$.
  Since $R$ has no facility to the left of $v$, the left side of $B'$ is to the right of $v$.
  Thus, $\dist(v, \text{left}(R)) \le \dist(\text{left}(B'),\text{left}(R))$.
  Observe that $r \ge \dist(\text{left}(B'),\text{left}(R)) + b' \ge 4 \cdot \dist(\text{left}(B'),\text{left}(R))$.
  It follows that
  \begin{equation}
    \label{eq:lem4_3}
    r/4 \ge \dist(v, \text{left}(R)).
  \end{equation}
  Combining Equations \ref{eq:lem4_1}, \ref{eq:lem4_2} and \ref{eq:lem4_3}, we obtain
  $$d_0 \ge r/3 - r/4 = r/12 \ge s/12.$$

\end{proof}

\begin{lemappx}[\ref{lem:fact1}]
  Let $\text{Event}_0(d, s)$ denote the event that an edge $e$ of length $d$ is separated by a cutting line of
  side-length $s$ that is produced by Cut-Rectangle.\\
  Then, $Pr[\text{Event}_0(d, s)] \le 3d/s.$
\end{lemappx}

\begin{proof}
  We consider the dissection tree. Let $R_0$ be the root of the dissection tree. 
  If $R_0$ has side-length $s$ then the probability that $e$ is cut by a line of side-length $s$ is 
  $3d/s$.
  Else, it does not matter whether $R_0$ cuts $e$ or not, and in any case we now look at the children of $R_0$ that contain $e$; say $R_1$ and $R_2$.
  If $R_1$ or $R_2$ has side-length $s$ then the probability that $e$ is cut by a line of side-length $s$ is then at most $3d/s$. Else, we go 
  deeper in the tree until we reach the rectangles that contain $e$ and have side-length $s$. 
  The probability that $e$ is cut by such a rectangle is thus at most $3d/s$.
  Hence, the probability that $e$ is cut by a line of side-length $s$ is at most $3d/s$.
\end{proof}

\section{Theorem \ref{thm:FL}}
\label{appx:FL}
\begin{thmappx}[\ref{thm:FL}]
  Algorithm \ref{algo:LS} produces a solution $L$ of cost at most $(1 + \mathcal{O}(\eps)) \cdot \cost(\opt)$.
\end{thmappx}

\begin{proof}

  Let $\opt$ be a globally optimum solution and $L$ be a locally optimum solution.
  By Proposition \ref{thm:new_struct}, for any $p>0$ there exists an assignment $E$ for each random dissection $\mathcal{D}_p$ with portals 
  of $L \cup \opt$,
  such that for any client $c$ and region $R$, if $c(L) \in R$ and $c(G) \notin R$ then 
  $c$ is served by a portal of $R$ or a facility of $L \setminus R$ in $E$ and
  the expected cost of $E$ is at most $\mathds{E} = \sum\limits_{c \in \mathcal{C}} \max(c_L,c_G) +  \mathcal{O}(\log(p)/p \cdot (\sum\limits_{c \in \mathcal{C}} (c_G + c_L)))$. \\
  This implies that there exists a dissection $\mathcal{D}_p$ for which $E$ has value at most $\mathds{E}$.

  Throughout the proof, we  consider this dissection $\mathcal{D}_p$ and fix $\eps := \log(p)/p$.
  Let $\mathcal{R}$ be the set of regions associated to $\mathcal{D}_{p}$.
  We start by constructing a solution $G$ based on $\opt$ and we  compare the cost of $L$ to the cost of $G$.
  The solution $G$ contains all the facilities of $\opt$ plus some extra facilities.
  First, it has a facility at each portal of $\mathcal{D}_{p}$.
  Moreover, for each region $R$ that is produced by the Partition Process, 
  we open the facilities of $L_R$. Recall that for each of these regions, $|L_R| \le 1/\eps $.
  We keep the same assignment for the clients.
  Since there are $\mathcal{O}(\eps^2 (|G|+|L|)$ regions and that for each region $G$ uses at most $1/\eps$ extra facilities, 
  the cost of $G$ is at most Cost($\opt$)+ $\mathcal{O}(\eps (|\opt| + |L|)f)$.
  We now prove that the cost of $L$ is at most $(1 + \mathcal{O}(\eps))/(1 - \mathcal{O}(\eps))$ times the cost of $G$, namely
  $$
  |L| \cdot f + \sum\limits_{c \in \mathcal{C}} c_L \le (\frac{1 + \mathcal{O}(\eps)}{1 - \mathcal{O}(\eps)}) (|G|\cdot f 
  + \sum\limits_{c \in \mathcal{C}} c_G).$$


  We focus on the cost of a region $R$. 
  We show that, by local optimality, for each region $R$, replacing solution $L$ by solution $G$ does not lead to a much better cost.
  We serve the clients of $C_R$ optimally (namely by the facilities that serve them in $G$) and
  the clients of $L_{R} \setminus G_R$ by the facilities located on the portals of $R$ or by the facilities
  of $L \setminus L_R$, depending on whether they belong to $C_L$ or $C_G$ and according to the assignment $E$.
  
  Since $|L_R \setminus G_R| + |G_R \setminus L_R|  = \mathcal{O}(\eps^3)$, the locality argument applies.
  Namely, we have
  $$
  (|G_R| - |L_R|)f + \sum\limits_{c \notin C_R \cup \Delta_R}c_L + \sum\limits_{c \in C_R} c_G + \sum\limits_{c \in \Delta_R} c_E 
  \ge (1-1/n)  (|L|f + \sum\limits_{c} c_L).
  $$
  Rearranging an summing over all region $R$ of $\mathcal{R}$, we derive

  \begin{equation}
    \label{eq:FL_comparative_cost}    
    \sum\limits_{R \in \mathcal{R}} \left( (|G_R| - |L_R|)f + \sum\limits_{c \in C_R} (c_G - c_L)  + 
    \sum\limits_{c \in \Delta_R} \left( c_E - c_L \right)  \right)\ge - \frac{|\mathcal{R}|}{n} \cdot \cost(L).  
  \end{equation}
  We now focus on proving an upper bound on the left-hand side of the above equation.
  We split the sum over $\Delta_R$ depending on whether $c$ is in $C_L$ or $C_G$.
  By Proposition \ref{thm:new_struct},
 $$\sum\limits_{\substack{c \in \Delta_{R}\\ c \in C_G}} 
 \left(c_E - c_L \right) \le \sum\limits_{\substack{c \in \Delta_R\\ c \in C_G}}  \left( c_L - c_L + \mathcal{O}(\eps \cdot (c_G + c_L)) \right),$$
 and

 $$ \sum\limits_{\substack{c \in \Delta_{R}\\ c \in C_L}} \left( c_E - c_L \right) \le   
  \sum\limits_{\substack{c \in \Delta_R\\ c \in C_L}} \left( c_G - c_L + \mathcal{O}(\eps \cdot (c_G+c_L))\right).$$
  Replacing in Inequality \ref{eq:FL_comparative_cost},
  $$   (|G| - |L|)f + \sum\limits_{R \in  \mathcal{R}}\left( \sum\limits_{c \in C_R} (c_G - c_L)  +  
  \sum\limits_{\substack{c \in \Delta_R\\ c \in C_L}} \left( c_G - c_L \right)\right) 
  + \sum\limits_{c \in \mathcal{C}} \mathcal{O}(\eps (c_G+c_L))
  \ge
  - \frac{\mathcal{|R|}}{n}  \cost(L).$$  
  Definition of $C_R$ leads to
  $$   (|G| - |L|)f + \sum\limits_{c \in \mathcal{C}} (c_G - c_L) + \mathcal{O}(\sum\limits_{c \in \mathcal{C}} \eps \cdot (c_G + c_L)) \ge
  - \frac{|\mathcal{R}|}{n} \cost(L).$$
  Since $|\mathcal{R}| = \mathcal{O}(\eps^2 \cdot n)$, we conclude 
  $$(1+ \mathcal{O}(\eps))\left(|G|f + \sum\limits_{c \in \mathcal{C}} c_G \right) \ge 
  (1 - \mathcal{O}(\eps))\left( |L|f + \sum\limits_{c \in \mathcal{C}} c_L \right)$$
  and the Theorem follows.
 
\end{proof}

\section{$k$-Median}
\label{appx:kmed}
\begin{lemappx}[\ref{lem:kmed_partition} Balanced Clustering]
  Let $\mathcal{R} = \{r_1,...,r_p\}$ be a collection of disjoint sets. Each set contains elements of type either $L$ or $G$
  and has size at least $1/2\eps^2$ and at most $1/\eps^2$. The total number of elements of type $L$
  is $(1+3\eps)$ times higher than the number of elements of type $G$.

  There exists a clustering of $\{r_1,...,r_p\}$ in clusters
  satisfying the following property. For any cluster $C$, 
  \begin{itemize}
  \item $C$ contains at most $\mathcal{O}(1/\eps^5)$ elements of $\mathcal{R}$, namely $|C| = \mathcal{O}(1/\eps^5)$;
  \item the difference between the number of elements of $L$ in the sets contained in $C$ 
    and the number of elements of $G$ in the sets contained in $C$ is at least $|C|/\eps$:
    $$\sum_{r_i\in C} |r_i \cap L| - \sum_{r_i\in C} |r_i \cap G| \geq |C|/\eps,$$
    for any $1/\eps \in \mathds{N}$.
  \end{itemize}
\end{lemappx}

\begin{proof}
  We first define for each set $r_i$, $v(r_i) := |L \cap r_i| - |G \cap r_i| - 1/\eps$.

  The assumption on the total number of elements of $L$ and $G$ can 
  be rewritten as $\sum\limits_{r_i} v(r_i) \ge \eps G > 0$.

  Besides, the cardinality bounds on $r_i$ imply that $v(r_i)$ is an 
  integer in the range $[-1/\eps^2-1/\eps, 1/\eps^2-1/\eps]$.

  We need to construct a clustering of $\mathcal{R}$ into small clusters
  such that for each cluster $C$,: $\sum_{r_i\in C} v(r_i)\geq 0$.
  We exhibit an algorithm that constructs such a clustering.
  For any set $r_i$ such that $v(r_i) = 0$, we create a new part that 
  contains only this set.
  This part trivially satisfies the above property.

  We now consider the remaining sets.
  While there exists $1 < i,j$, such that $\lemThreshold < |v_i|,|v_{-j}|$,
  We take $i$ sets from $v_{-j}$ and $j$ sets from $v_i$ and create a new part that contains them all.
  This part satisfies the property of the Lemma and contains at most $2/\eps^2$ sets of $\mathcal{R}$.
  
  We now turn to the last case, namely $\forall j \ge 0$, $|v_j| \le \lemThreshold$ 
  (or symmetrically $\forall j \le 0$, $|v_j| \le \lemThreshold$).
  We claim that it is possible to make on last part containing all the remaining sets and that this part
  satisfies the property of the Lemma and has size $\mathcal{O}(1/\eps^5)$.
  We start by proving that, after each step $s$ of the above algorithm, 
  the following invariant holds
  \begin{equation}
    \label{eq:lem_IH}
    (1+\frac{2\eps}{1-\eps})|G_s|\le |L_s| \le (1+\frac{4\eps}{1-2\eps}) |G_s|,
  \end{equation}
  where $L_s$ and $G_s$ are the number of elements of type $L$ and $G$ respectively that
  are not contained in any part after step $s$.\\
  This is true at the beginning of the algorithm. We show that it is true all the way to the last step.
  Assume that it holds after step $s$, we prove that it is true after step $s+1$.\\
  Let $P$ be the part created at step $s$. This part contains say $P_G$ elements of $G$ and so, $P_G + |P|/\eps$ elements of 
  $L$.
  By induction hypothesis, Inequality \ref{eq:lem_IH} holds.
  Hence, by expressing $L_{s-1}$ and $G_{s-1}$ in terms of $L_s$ and $G_s$, it follows that
  $$(1+\frac{2\eps}{1-\eps})(|G_s| + P_G) \le L_s + P_G + |P|/y \le (1+\frac{4\eps}{(1-2\eps)})(|G_s| + P_G).$$
  By definition of the $s_i$, 
  $$|P|/2\eps^2 \le 2P_G + |P|/\eps \le |P|/\eps^2.$$
  Rearranging and replacing in the inequalities above, it follows 
  $$\frac{2\eps}{1-\eps} P_G \le |P|/\eps \le \frac{4\eps}{1-2\eps} P_G.$$

  At final step $f$, the upper and lower bounds on $L_f$ induced by Inequality \ref{eq:lem_IH} implies that 
  the final part has size at most $\mathcal{O}(1/\eps^5)$ and satisfies the properties of the Lemma.
\end{proof}

\begin{thmappx}[\ref{thm:kmed}]
  Algorithm \ref{algo:LS} for $k$-Median produces a solution $L$ that is a $(1 + \mathcal{O}(\eps), 1 + \mathcal{O}(\eps))$ bi-criteria approximation.
\end{thmappx}

\begin{proof}
  Remark first that solution $L$ uses $(1+\mathcal{O}(\eps))k$ facilities.
  We now show that the cost of solution $L$ is at most 
  $1 + \mathcal{O}(\eps)$ times higher than the cost of the optimal solution.\\
  Recall that by Proposition \ref{thm:new_struct}, for any $p>0$ there exists an assignment $E$ for each random dissection $\mathcal{D}_p$ 
  of $L \cup \opt$ with portals, such that for any client $c$ and region $R$, if $c(L) \in R$ and $c(\opt) \notin R$ then 
  $c$ is served by a portal of $R$ or a facility of $L \setminus R$ in $E$ and
  the expected cost of $E$ is at most $\mathds{E} = \sum\limits_{c \in \mathcal{C}} \max(c_L,c_\opt) +  \mathcal{O}(\log(p)/p \cdot (\sum\limits_{c \in \mathcal{C}} (c_\opt + c_L)))$. \\
  This implies that there exists a dissection $\mathcal{D}_p$ for which $E$ has value at most $\mathds{E}$.
  Throughout the proof, we  consider such a dissection $\mathcal{D}_p$ and fix $\eps := \log(p)/p$.
  Let $\mathcal{R}$ be the set of regions associated to $\mathcal{D}_{p}$.

  We prove that the cost of $L$ is at most $(1 + \mathcal{O}(\eps))/(1 - \mathcal{O}(\eps))$ times the cost of $S$, namely
  $$
  \sum\limits_{c \in \mathcal{C}} c_L \le \frac{1 + \mathcal{O}(\eps)}{1 - \mathcal{O}(\eps)} \sum\limits_{c \in \mathcal{C}} c_\opt.
  $$


  Let $\mathcal{P}$ be a clustering of the the regions satisfying the properties of Lemma \ref{lem:kmed_partition} 
  (depending on $L$ and $\opt$).
  We start by constructing a solution $G$ based on $\opt$ and we  compare the cost of $L$ to the cost of $G$.
  We construct $G$ in a similar way to in the proof of Theorem \ref{thm:FL}.
  Namely, the solution $G$ contains all the facilities of $\opt$ plus some extra facilities:
  one facility at each portal of $\mathcal{D}_{p}$ and for each region $R$ that is produced by the Partition Process, 
  we open the facilities of $L_R$. Recall that for each of these regions, $|L_R| \le 1/\eps $.
  We keep the same assignment for the clients.
  
  We now compare the costs of $L$ and $G$. To do so, we  consider all the regions of each cluster of the clustering $\mathcal{P}$
  at the same time. Namely for each cluster $R$, $L$ uses at least as many facilities as $G$.
  Therefore $|S_P \setminus L| + |L \setminus S_P| = \mathcal{O}(1/\eps^9)$ and the locality argument applies.
  


  We show that, by local optimality, the cost of $S_P$ is close to
  the cost of $L$.
  We serve the clients of $C_R$ optimally (namely by the facilities that serve them in $G$) and
  the clients of $\Delta_R$ by the facilities located on the
  portals of $R$ or by the facilities
  of $L \setminus L_R$, according to the assignment $E$.
  By local optimality, the cost of replacing $L$ by $S_P$ is greater (up to a factor $(1-1/n)$) than the cost of $L$.
  Namely, we have
  $$
    \sum\limits_{R \in P} \left( \sum\limits_{c \notin C_R \cup \Delta_R}c_L + \sum\limits_{c \in C_R} c_G + \sum\limits_{c \in \Delta_R} c_E \right)
    \ge (1-1/n)  \cost(L).
  $$

  Rearranging and  summing over all part $P$ of $\mathcal{P}$,

  \begin{equation}
    \label{eq:kmed_comparative_cost}    
    \sum\limits_{P \in  \mathcal{P}} \sum\limits_{R \in P} \left( \sum\limits_{c \in C_R} (c_G - c_L)  + 
    \sum\limits_{c \in \Delta_R} \left( c_E - c_L \right)  \right)\ge - \frac{|\mathcal{P}|}{n} \cdot \cost(L).  
  \end{equation}

  We now provide an upper bound on the left-hand side of the above equation. 
  We separate the sum over $\Delta_R$ depending on whether $c$ is in $C_L$ or $C_G$.

  By Proposition \ref{thm:new_struct}, we obtain
  $$\sum\limits_{\substack{c \in \Delta_{R}\\ c \in C_G}} 
  \left(c_E - c_L \right) \le \sum\limits_{\substack{c \in \Delta_R\\ c \in C_G}}  \left( c_L - c_L + \mathcal{O}(\eps \cdot (c_G + c_L)) \right),$$
  and

  $$ \sum\limits_{\substack{c \in \Delta_{R}\\ c \in C_L}} \left( c_E - c_L \right) \le   
  \sum\limits_{\substack{c \in \Delta_R\\ c \in C_L}} \left( c_G - c_L + \mathcal{O}(\eps \cdot (c_G+c_L))\right).$$

  Replacing in Equation \ref{eq:kmed_comparative_cost}, it follows that 
  $$ \sum\limits_{P \in  \mathcal{P}}\sum\limits_{R \in P}\left( \sum\limits_{c \in C_R} (c_G - c_L)  +  
  \sum\limits_{\substack{c \in \Delta_R\\ c \in C_L}} \left( c_G - c_L \right) \right)
  + \sum\limits_{c \in \mathcal{C}} \mathcal{O}(\eps \cdot (c_G+c_L)) 
  \ge - \frac{|\mathcal{P}|}{n} \cdot \cost(L).  $$
  
  By the definition of $C_R$, the left-hand side is exactly
  $$ \sum\limits_{c \in \mathcal{C}} (c_G - c_L) +\sum\limits_{c \in \mathcal{C}} 
   \mathcal{O}( \eps \cdot (c_G + c_L)).$$
  
  Since $|\mathcal{P}| = \mathcal{O}(\eps k)$, we conclude 
  $$(1+ \mathcal{O}(\eps)) \cdot \sum\limits_{c \in \mathcal{C}} c_G \ge 
  (1 - \mathcal{O}(\eps)) \sum\limits_{c \in \mathcal{C}} c_L$$
  and the Theorem follows.

\end{proof}

\end{document}